\documentclass[a4paper,UKenglish,cleveref, autoref]{lipics-v2019}
\usepackage{graphicx}
\usepackage[utf8]{inputenc}
\usepackage{xspace}
\usepackage{thmtools}

\newcommand{\polydots}{\emph{Dots \& Polygons}\xspace}
\newcommand{\boxdots}{\emph{Dots \& Boxes}\xspace}
\newcommand{\simpledots}{\emph{Dots \& Simple Polygons}\xspace}
\newcommand{\holesdots}{\emph{Dots \& Polygons \& Holes}\xspace}
\newcommand{\B}{\ensuremath{\mathcal{B}}\xspace}
\newcommand{\R}{\ensuremath{\mathcal{R}}\xspace}

\newcommand{\SBf}{\ensuremath{S_{\mathit{Bfinal}}}\xspace}
\newcommand{\SRf}{\ensuremath{S_{\mathit{Rfinal}}}\xspace}
\newcommand{\Nchains}{\ensuremath{N_{\mathit{chains}}}\xspace}
\newcommand{\Ncycles}{\ensuremath{N_{\mathit{cycles}}}\xspace}
\newcommand{\Nunclaimed}{\ensuremath{N_{\mathit{unclaimed}}}\xspace}
\newcommand{\Abell}{\ensuremath{A_{\mathit{bell}}}\xspace}

\title{Dots \& Polygons}

\author{Kevin Buchin}{Department of Mathematics and Computer Science, TU Eindhoven, Netherlands}{k.a.buchin@tue.nl}{https://orcid.org/0000-0002-3022-7877}{}

\author{Mart Hagedoorn}{Department of Mathematics and Computer Science, TU Eindhoven, Netherlands}{}{}{} 

\author{Irina Kostitsyna}{Department of Mathematics and Computer Science, TU Eindhoven, Netherlands}{i.kostitsyna@tue.nl}{https://orcid.org/0000-0003-0544-2257}{}

\author{Max van Mulken}{Department of Mathematics and Computer Science, TU Eindhoven, Netherlands}{}{}{} 

\author{Jolan Rensen}{Department of Mathematics and Computer Science, TU Eindhoven, Netherlands}{}{}{} 
 
\author{Leo van Schooten}{Department of Mathematics and Computer Science, TU Eindhoven, Netherlands}{}{}{} 

\authorrunning{K. Buchin, M. Hagedoorn, I. Kostitsyna, M. v. Mulken, J. Rensen, L. v. Schooten}

\Copyright{Kevin Buchin, Mart Hagedoorn, Irina Kostitsyna, Max van Mulken, Jolan Rensen, Leo van Schooten}



\relatedversion{}

\supplement{game: \url{https://www.win.tue.nl/~kbuchin/proj/ruler/dotsandpolygons/}, source code: \url{https://github.com/kbuchin/ruler/tree/dots-and-polygons}}

\acknowledgements{We would like to thank David Eppstein for discussing~\cite{Epp} with us.}

\hideLIPIcs  

\EventEditors{John Q. Open and Joan R. Acces}
\EventNoEds{2}
\EventLongTitle{42nd Conference on Very Important Topics (CVIT 2016)}
\EventShortTitle{CVIT 2016}
\EventAcronym{CVIT}
\EventYear{2016}
\EventDate{December 24--27, 2016}
\EventLocation{Little Whinging, United Kingdom}
\EventLogo{}
\SeriesVolume{42}
\ArticleNo{23}

\begin{document}
\nolinenumbers
\maketitle

\begin{abstract}
We present a new game, \polydots, played on a planar point set. Players take turns connecting two points, and when a player closes a (simple) polygon, the player scores its area. We show that deciding whether the game can be won from a given state, is NP-hard. We do so by a reduction from vertex-disjoint cycle packing in cubic planar graphs, including a self-contained reduction from planar 3-Satisfiability to this cycle-packing problem. This also provides a simple proof of the NP-hardness of the related game \boxdots. For points in convex position, we discuss a greedy strategy for \polydots.
\end{abstract}

\begin{figure}[h]
    \centering
    \includegraphics[width=\textwidth]{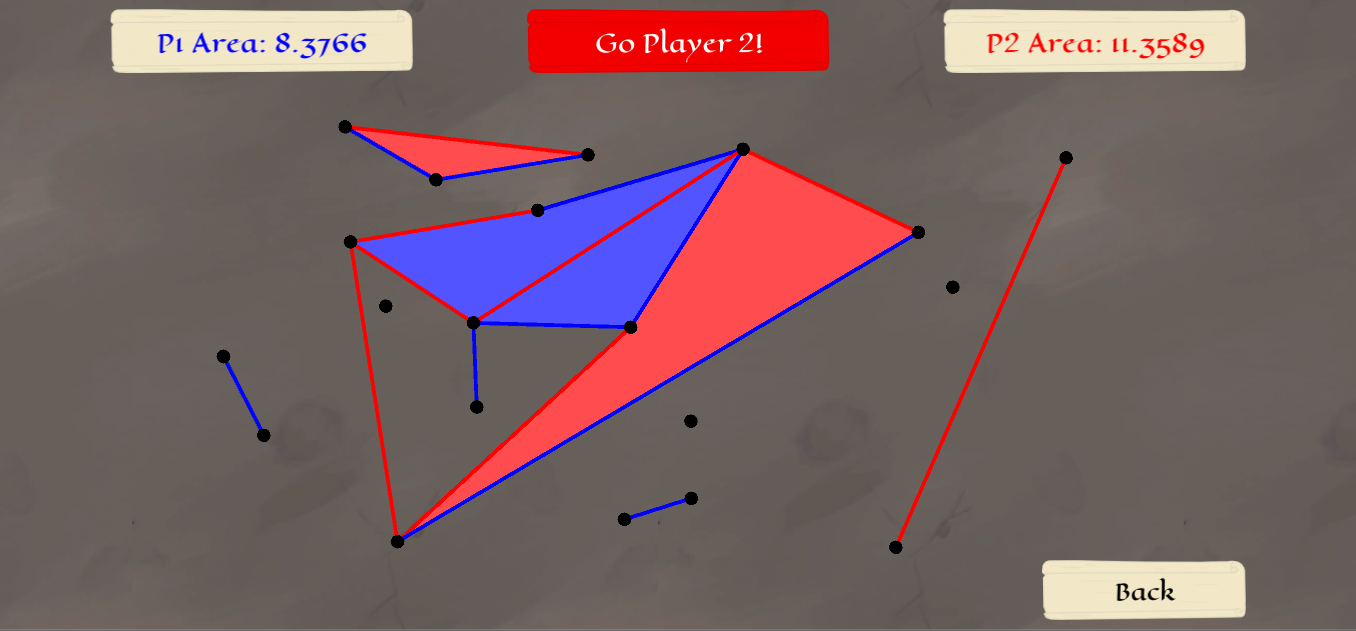}
    \caption{A screenshot of an intermediate game state in the \simpledots game.}
    \label{fig:screenshot}
\end{figure}

\section{Introduction}
\label{introduction}

\boxdots~\cite{berlekamp2000dots} is a popular game, in which two players take turns in connecting nodes lying on the integer lattice, scoring when they surround unit squares. We introduce a more geometric variant of this game: \polydots.

The game is played on a planar point set $P$ of size $n$. Two players, \R and \B (player \R always goes first), take turns connecting two points $p,q \in P$ by a straight-line edge in a turn.
The edge may not intersect other points or edges, and may not lie in a previously scored area.
When a player closes a polygon, they score the area of the polygon 
 and must make another move. When all area of the convex hull of $P$ has been scored, the player with the larger total area wins.

 We distinguish two variants of the game. 
  In \holesdots, when a player closes a cycle, the player scores the enclosed area (excluding possibly previously enclosed parts). In \simpledots, a player only scores when they close a simple polygon with no points inside. Figure~\ref{fig:variants} illustrates the difference between the variants.
 
\begin{figure}[t]
    \centering
    \includegraphics{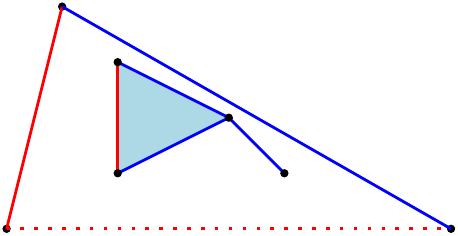}
    \caption{In \holesdots, if \R draws the dotted edge, they will score the interior minus the blue triangle. In \simpledots, \R will not score in this way in this turn, since the enclosed polygon is not simple.}
    \label{fig:variants}
\end{figure}

A similar game is \emph{Monochromatic Complete Triangulation Game}~\cite{aichholzer}, but in that game only triangles are scored, and the score is the number of triangles, rather than the area. We note that there is another variant of \boxdots also called Dots \& Polygons~\cite{DP-first} that is played on the integer lattice. 

\subparagraph*{Contributions.}
We implemented \polydots in the \emph{Ruler of the Plane} framework~\cite{ruler17}.
Both variants of the game can be played online (see supplementary materials).
 In Section~\ref{hardness} we prove that 
deciding whether \simpledots can be won from a given state is NP-hard. 
We do so by a reduction from vertex-disjoint cycle packing in cubic planar graphs, including a self-contained reduction from planar 3-Satisfiability to this cycle-packing problem, and from this cycle-packing problem to \boxdots. In Section~\ref{strategy} we discuss a greedy strategy for the case that $P$ is in convex position.

\section{Hardness}
\label{hardness}

In this section we show that \simpledots is NP-hard by a reduction from the maximum cycle packing problem in planar cubic graphs.
The reduction is similar to the proof of NP-hardness of \boxdots.
The book \emph{Winning Ways for your Mathematical Plays}~\cite{winning_ways} mentions that a generalization of \boxdots can be shown to be NP-hard by a reduction from the maximum vertex-disjoint cycle packing (VCP) problem.
The VCP problem can be viewed as a generalization of the triangle packing problem~\cite{Bodlaender92}, which is known to be NP-hard~\cite{Garey79}.

Eppstein notes that the NP-hardness, mentioned in~\cite{winning_ways}, should apply to the classic \boxdots by a reduction from the VCP problem in planar cubic graphs~\cite{Epp}. 
However, he does not cite a source of the hardness proof for this VCP variant.
Furthermore, triangle packing is polynomial-time solvable in planar graphs with maximum degree three~\cite{Caprara2002}, and thus can no longer be used to justify the hardness of the VCP in planar cubic graphs.
Thus, for the sake of completeness, we also show the following two theorems.

\begin{theorem}
\label{theo:vertexDisCycles}
Maximal vertex-disjoint cycle packing in planar cubic graphs is NP-complete.
\end{theorem}
\begin{proof}
We reduce from the planar 3-Satisfiability~\cite{Lichtenstein82}.
Consider an instance of the planar 3-Satisfiability problem with $n$ variables $X=\{x_1,\dots,x_n\}$ and $m$ clauses $\mathcal{C}=\{C_1,\dots,C_m\}$.
We construct a graph $G$, corresponding to the 3-Satisfyability instance, in which a certain number of vertex disjoint cycles exists if and only if there exists an assignment of true/false values to the variables in $X$ such that all clauses in $\mathcal{C}$ are satisfied.

\begin{figure}[t]
\centering
\includegraphics{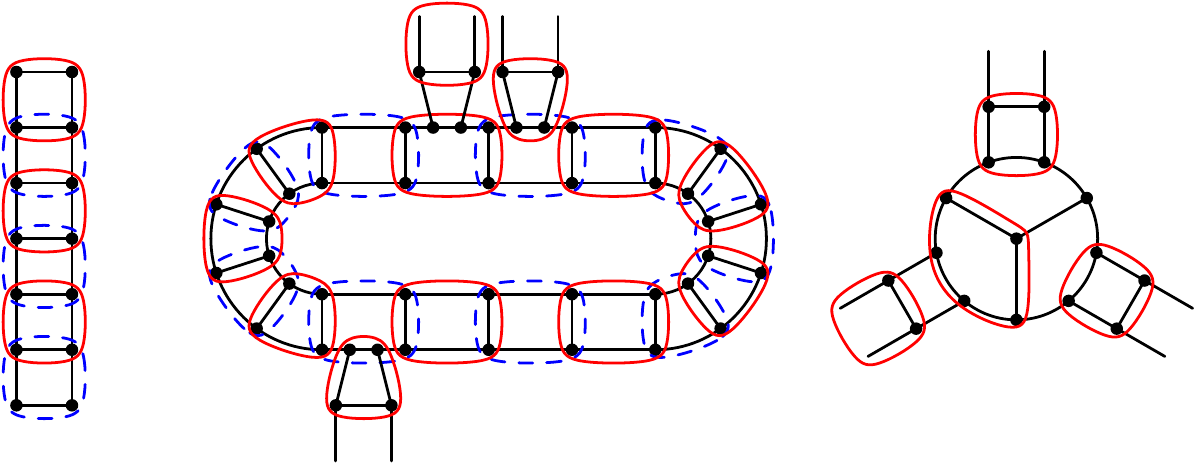}
\caption{Left: Wire gadget. Either blue or red cycles will be chosen; Middle: Vertex gadget. Red cycles correspond to the variable set to true, and blue---to false. Wire gadgets connect the variable to the corresponding clause gadgets; Right: clause gadget. To select at least one cycle from the middle of the gadget, at least one of the three branches attached has to end with a non-selected $4$-cycle.}
\label{fig:cycle-packing}
\end{figure}

A wire gadget, shown in Figure~\ref{fig:cycle-packing} (left), consists of a chain of $2k_s$ $4$-cycles for some value of $k_s$.
If at least $k_s$ cycles have to be chosen from the gadget, an alternating order of $4$-cycles must be selected.

For each vertex we create a cycle of $2k_v$ $4$-cycles for some value of $k_v$ (refer to Figure~\ref{fig:cycle-packing} (middle)), such that either all odd $4$-cycles or all even $4$-cycles must be chosen if we want to select at least $k_v$ cycles from the gadget.
One choice of alternating $4$-cycles will correspond to setting the variable to true, and the other---to false.

For each clause we create a clause gadget shown in Figure~\ref{fig:cycle-packing} (right).
The three cycles in the middle pairwise share an edge.
Each of them is connected by a wire gadget to a corresponding vertex gadget.
Thus, to select at least one cycle from the middle, one of the three wire chains attached must end with a non-selected $4$-cycle.

Let $2K_s$ be the total number of $4$-cycles used to construct the wire gadgets, $2K_v$ be the total number of $4$-cycles used to construct the variable gadgets.
The truth assignment to the $m$ clauses in $\mathcal{C}$ of the planar 3-Satisfiability instance corresponds to a selection of $K_s+K_v+m$ vertex-disjoint cycles in the resulting graph.
Note, that any selection of the cycles other than specified will lead to the number of vertex-disjoint cycles strictly less than $K_s+K_v+m$.

Thus vertex-disjoint cycle packing in planar cubic graphs is NP-hard. The problem is in NP, as the correct cycle selection can be verified in polynomial time.
\end{proof}

\begin{theorem}
\label{theo:dotsboxes}
Given a \emph{game state} of \boxdots, it is NP-hard to decide whether \B can win.
\end{theorem}
\begin{proof}
Given a planar cubic graph $G$, we can construct a \emph{game state}\footnote{We define a \emph{game state} to be a specific configuration of a game board of \boxdots or \simpledots game. In it players \R and \B have a score corresponding to the amount of area that they have already captured.} of \boxdots with it being player \B's turn, such that \B can only win if they find the maximum number of vertex-disjoint cycles in $G$. Consider an orthogonal embedding $G'$ of $G$ on a grid. We place extra degree-2 vertices on each edge of $G$ such that a chain of at least four vertices exists between each pair of degree-3 vertices (Figure~\ref{fig:reduction} (left)).

In order to construct the \emph{game state}, we consider every vertex of $G'$ as a \boxdots cell and surround all edges by walls. This way each cell is only open along the sides corresponding to the incident edges of the vertex (Figure~\ref{fig:reduction} (middle)). 

\begin{figure}[t]
    \centering
    \includegraphics[page=4]{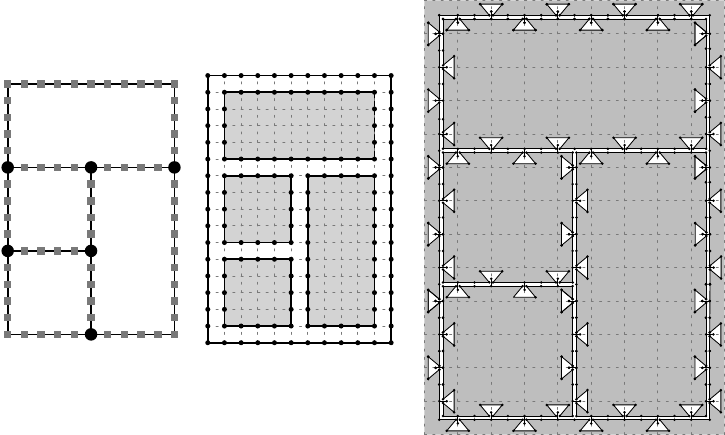}
    \caption{Left: an orthogonal embedding $G'$ of $G$ with extra nodes added to make long chains; Middle: corresponding \boxdots instance; Right: corresponding \simpledots instance.}
    \label{fig:reduction}
\end{figure}

In this state of the game any move that player \B can make allows \R to subsequently close off a sequence of boxes. We then say that in such a game state player \R is in \emph{control} of the game. In order for \R to retain \emph{control} of the game, they can perform a so called \emph{double-cross} move; after claiming part of a chain/cycle, they leave one (in case of a chain) or two (in case of a cycle) unclaimed $2\times 1$ rectangles (refer to Figure~\ref{fig:double-cross}). After such a move, player \B can claim these open rectangles, but \R will stay in \emph{control} of the game.

Let $S_B$ be the score gained by \B and $S_R$ be the score gained by \R leading up to the \emph{game state} represented by the construction based on graph $G'$. Note that the sum of these two scores is equal to the total number of claimed boxes. Since \R is in \emph{control} of the game in this \emph{game state}, \B is forced to continually opening chains and cycles for \R to claim. This means, however, that \B is able to divide the remaining unclaimed boxes into chains and cycles independent of \R. Assuming \R plays optimally, we can compute the final score $\SBf$ of \B using the following formula:
\[ \SBf = S_B + 2(\Nchains-1) + 4\Ncycles\,, \]
where $\Nchains$ is the number of chains and $\Ncycles$ is the number of cycles claimed by \R. If playing optimally, \R will make a double-cross move in all but the last chain/cycle they can close, which will give \B two boxes in the case of a chain of boxes, and four boxes in the case of a cycle of boxes. Since the area gained by \B is larger when a double-cross move is done in a cycle rather than in a chain, it is beneficial for \B to aim towards a \emph{game state} in which the number of cycles is maximized. Thus, in order to maximize their score, \B needs to find the maximum number of vertex-disjoint cycles in $G$. 

Then, in order to decide whether \B can win from a given \emph{game state} we need to find the difference between \SBf and \SRf, where \SRf is equal to the final score of \R, calculated as follows:
\[ \SRf = S_R + (\Nunclaimed - 2(\Nchains-1) - 4\Ncycles)\,, \]
where $\Nunclaimed$ is the number of unclaimed boxes in the \emph{game state} represented by the transformation from $G$. To decide whether \B can win we simply compare the maximized \SBf with \SRf. Thus, since the intermediate step of finding the maximized final score of \B is NP-hard, as proven in Theorem \ref{theo:vertexDisCycles}, it is NP-hard to decide whether \B can win.
\end{proof}

We now show a similar reduction for \simpledots.
\begin{theorem}
Given a state of \simpledots, it is NP-hard to decide whether \B can win.
\end{theorem}
\begin{proof}
Starting in a planar cubic graph $G$, we first construct an instance of \boxdots as described in the proof of Theorem \ref{theo:dotsboxes} where $G'$ is defined as in Theorem \ref{theo:dotsboxes}. Using this instance of \boxdots we can construct a \simpledots instance by considering each cell as a corridor connecting its open sides, and subsequently adding a bell-shaped area on each edge of $G'$ connecting two vertices, as shown in Figure \ref{fig:gadgets}. Such a construction for graph $G$ is shown in Figure \ref{fig:reduction} (right).

\begin{figure}[t]
    \centering
    \includegraphics[page=5]{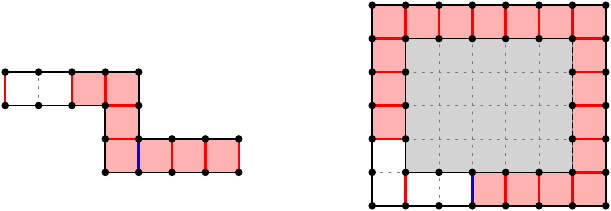}
    \caption{Double-cross move by \R. If \B opens a chain (or a cycle), \R can claim a sequence of boxes. To pass the turn back to \B, \R can leave two (or four) boxes unclaimed.}
    \label{fig:double-cross}
\end{figure}

The resulting \emph{game state} yields a situation in which \R is in \emph{control} of the game; any move \B can make allows \R to subsequently close off a sequence of simple polygons. Since the winning condition of \simpledots is in terms of total amount of area and not the total amount of polygons claimed, it is important to note that the area of all the bell-shaped polygons are equal. If \R plays optimally they would want to stay in \emph{control} of the game. Again, this means that they will play a double-cross move in all but the last chain/cycle they can close, meaning \B has to open up new cycles or chains containing more bell-shaped polygons for \R. The concept of a double-cross move in \simpledots is similar to a double-cross move in \boxdots (Figure \ref{fig:double-cross-poly}). Such a move will give two bell shaped areas to \B in a cycle and one bell shaped area in a chain.

Similarly to the proof of Theorem \ref{theo:dotsboxes} let $\Abell$ be the area of a bell shaped subpolygon, let $S_B$ be the score gained by \B and $S_R$ be the score gained by \R. The scores $S_R$ and $S_B$ lead up to the \emph{game state} represented by the transformation from graph $G'$ (Figure \ref{fig:reduction} (right)). Note that the sum of these two scores is equal to the area of the grey regions in Figure \ref{fig:reduction}. Assuming \R plays optimally, we can calculate the final score \SBf of \B:
\[ \SBf = S_B + (\Nchains-1) \Abell + 2 \Ncycles \Abell\,, \]
where \Nchains is the number of chains and \Ncycles is the number of cycles claimed by \R.
If playing optimally, \R will make a double-cross move in all but the last chain/cycle they can close, which will give \B one bell-shaped area in the case of a chain, and two bell-shaped areas in the case of a cycle. Since the area gained by \B is larger when a double-cross move is done in a cycle rather than in a chain, it is beneficial for \B to aim towards a \emph{game state} in which the number of cycles is maximized.
Thus, in order to maximize their score, \B needs to find the maximum number of vertex-disjoint cycles in $G$. 

Then, in order to decide whether \B can win from a given \emph{game state} we need to find the difference between \SBf and \SRf, where \SRf is equal to the final score of \R:
\[ \SRf = S_R + (\Nunclaimed - (\Nchains-1) \Abell - 2 \Ncycles \Abell)\,, \]
where $\Nunclaimed$ is amount of unclaimed area in the game state represented by the transformation from $G$. To decide whether \B can win we simply compare the maximized \SBf with \SRf. Thus, since the intermediate step of finding the maximized final score of \B is NP-hard, as proven in Theorem \ref{theo:vertexDisCycles}, it is NP-hard to decide whether \B can win.
%
\end{proof}

\begin{figure}[t]
    \centering
    \includegraphics[page=1]{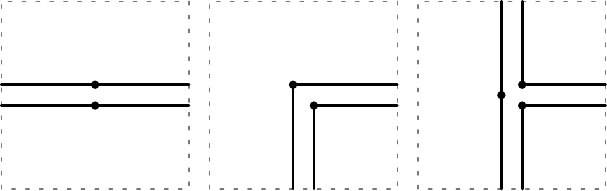}\hfil
    \includegraphics[page=2]{figures/dots-polygons-np-hard.pdf}
    \caption{Gadgets: vertices (left) and edges (right). A bell-shaped area is added to the middle of each edge of $G'$ connecting two vertices. Point $a$ does not see point $c$.}
    \label{fig:gadgets}
\end{figure}

\begin{figure}[t]
    \centering
    \includegraphics[page=3]{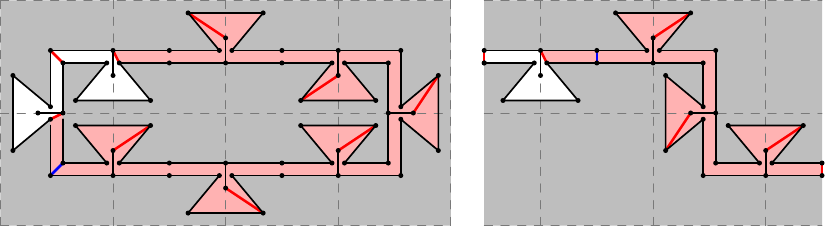}
    \caption{Double-cross move by \R after \B draws any diagonal in a cycle (left) or a chain (right). \R gives away two bell-shaped areas to \B in a cycle, and one bell-shaped area in a chain.}
    \label{fig:double-cross-poly}
\end{figure}

\section{Strategy}
\label{strategy}

In the following we discuss greedy strategies for \polydots played on a set of points $P$ in convex position. In this case, both variants of the game are the same. 
In the related \emph{Monochromatic Complete Triangulation Game} a greedy strategy is optimal for such points~\cite{aichholzer}.

We first observe that the number of turns is always the same.
\begin{lemma}
The number of turns in a game of \polydots played on a set of points in convex position is equal to $n = |P|$.
\end{lemma}
\begin{proof}
Consider connected components of the edges drawn by the players.
If a player connects two points in the same component, this closes a polygon, and therefore the turn continues. If, however, the two points are in different components, the turn ends and the number of connected components decreases. Thus, the number of turns equals to the number of initial components.
\end{proof}

Consider a \emph{game state} in which the current player cannot close a polygon. Let $E$ be the set of all edges that can still be drawn. 
Define the weight $w(e)$ for $e \in E$ to be the area the opponent can claim on their next turn if the current player draws $e$. For example an edge $e$ between two isolated points has weight $w(e) = 0$.
A \emph{simple greedy strategy} is the following: if there is an edge that can close some area, immediately draw that edge. Otherwise, draw an arbitrary edge $e_{min} = \min\limits_{\forall e \in E} w(e)$. This strategy is not optimal, as shown in Figure~\ref{fig:bad_case}.
\begin{figure}[t]
    \centering
    \includegraphics{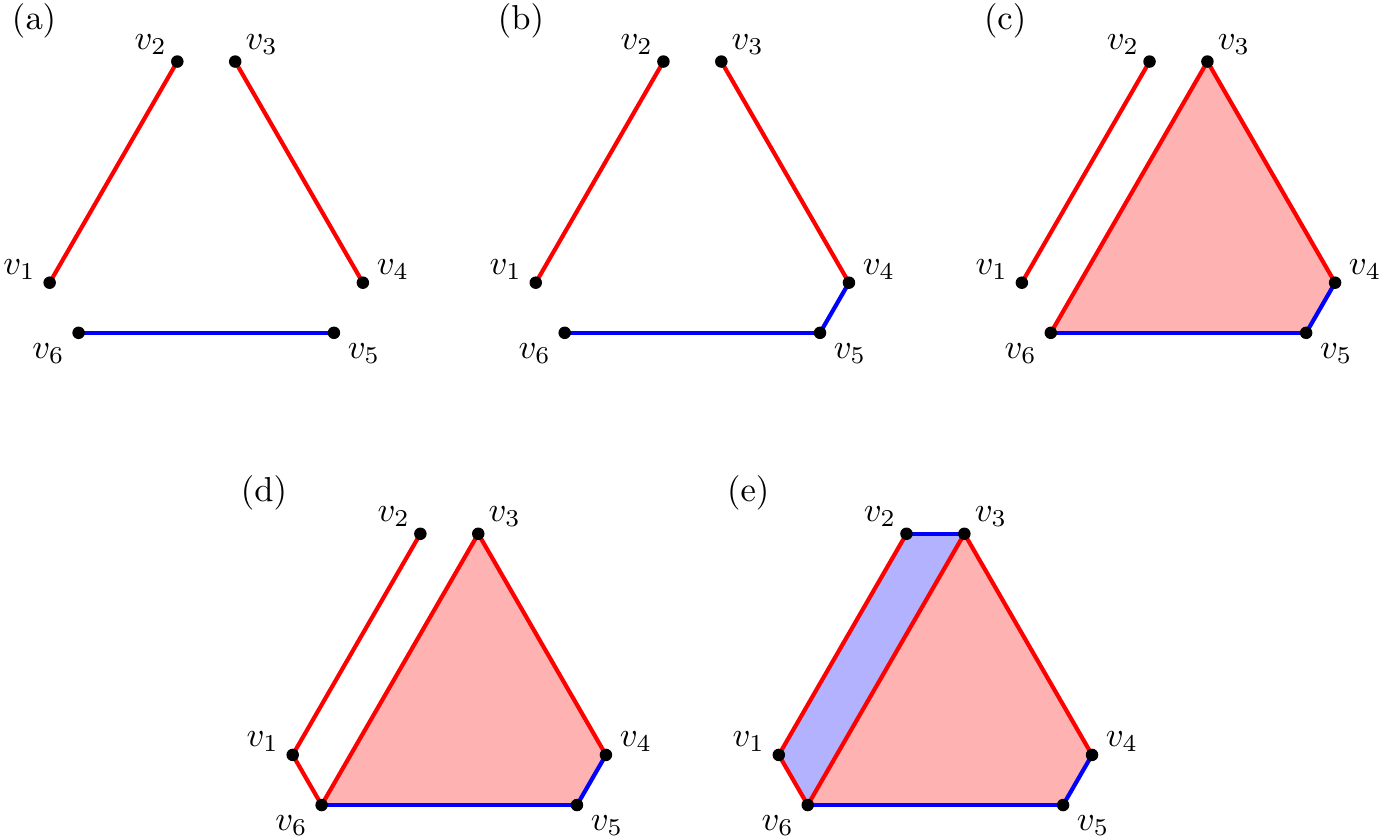}
    \caption{Player \B could have won this game, but after reaching the state in (a) loses as shown in (b--e).}
    \label{fig:bad_case}
\end{figure}

The edges drawn partition the remaining area into \emph{subproblems}. For an edge $e \in E$, $w(e)$ can only change if an edge in the same subproblem is drawn. Let $E'\subset E$ be the set of edges within a subproblem. We call a subproblem \emph{easy}, if only two of the edges $e, e' \in E'$ lie on the convex hull of $P$.
In such a subproblem, all edges have the same weight, namely the area of the subproblem. We call a \emph{game state} in which all subproblems are easy, an \emph{easy endgame}.

We will show that a player that goes last can win the game following the simple greedy strategy, if they can enforce an easy endgame.
As an example, consider an easy endgame state on a set of $n=6$ points shown in Figure~\ref{fig:exampleGame}~(a), which \B can force by drawing the diagonal $v_2 v_5$ (or an equivalent diagonal up to symmetry) in their first move, if it does not yet exist.
Recall that player \R goes first, and in this specific case \B goes last.
After the first three moves shown in the figure, it is \B's turn again.
Player \B has two choices, either draw an edge that gives away the quadrilateral $\Box v_1 v_2 v_5 v_6$, or give away the quadrilateral  $\Box v_2 v_3 v_4 v_5$.
An edge with the minimal weight would be the edge that gives away the least amount of area. In this case that would be either edge $v_2 v_3$, $v_2 v_4$, $v_3 v_5$ or $v_1 v_4$.
In the next turn \R can claim the smaller quadrilateral area, and then \R must give away the larger quadrilateral area to \B, resulting in a win for \B (refer to Figure~\ref{fig:exampleGame}~(b--e)).

\begin{figure}[t]
    \centering
    \includegraphics{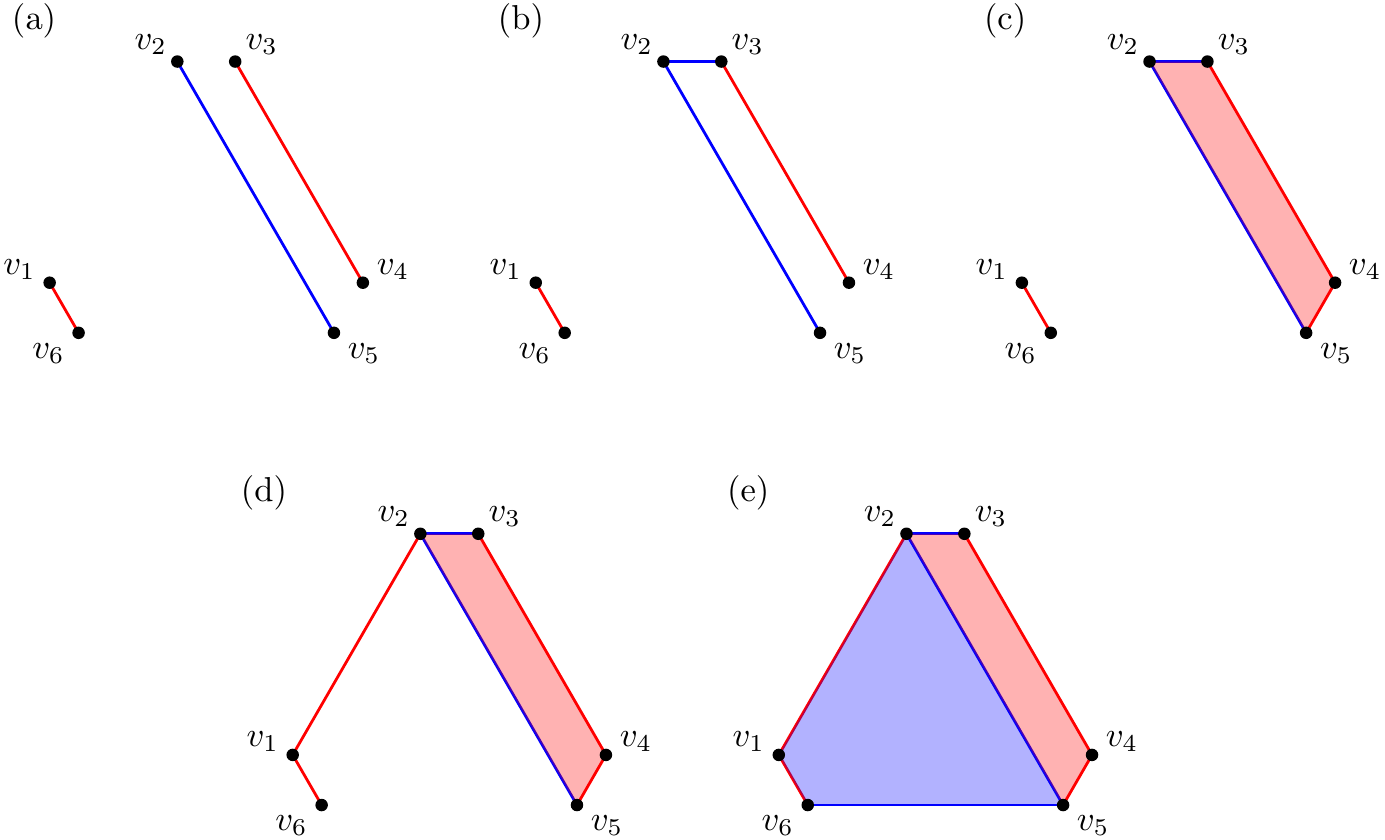}
    \caption{Player \B wins by following a simple greedy strategy.}
    \label{fig:exampleGame}
\end{figure}

\begin{theorem}
Let $P$ be a set of $n$ points in the plane in convex position. For $n = 3,5,7$ player \R can score at least half of the area, for $n =4,6$ player \B can score at least half of the area.
\end{theorem}
\begin{proof}
Consider the player that will go last (i.e., \R for odd $n$, \B for even $n$). If this player plays the \emph{simple greedy strategy} in such a way that they reach an easy endgame, then they win. Indeed, from that point onward, any time the opponent scores an area $A$, the current player will score an area that is at least as large as $A$ in their next turn. 

For $n \leq 4$, an easy endgame is always reached. For $5 \le n \le 7$, the player that will go last can enforce an easy endgame by playing a diagonal in their first move, preventing a non-easy endgame, e.g., as in Figure~\ref{fig:bad_case}.
Thus, \R can always win for $n=3,5,7$, and \B for $n=4,6$. 
\end{proof}
Thus, if $P$ is chosen such that a draw is not possible, the player that goes last wins the game. We leave the problem for $n>7$ open.

\enlargethispage{\baselineskip}



\bibliography{refs}

\begin{thebibliography}{10}

\bibitem{aichholzer}
Oswin Aichholzer, David Bremner, Erik~D. Demaine, Ferran Hurtado, Evangelos
  Kranakis, Hannes Krasser, Suneeta Ramaswami, Saurabh Sethia, and Jorge
  Urrutia.
\newblock Games on triangulations.
\newblock {\em Theoretical Computer Science}, 343(1–2):52--54, 2005.

\bibitem{ruler17}
Sander Beekhuis, Kevin Buchin, Thom Castermans, Thom Hurks, and Willem Sonke.
\newblock Ruler of the plane -- games of geometry (multimedia contribution).
\newblock In {\em 33rd International Symposium on Computational Geometry
  (SoCG)}, volume~77 of {\em LIPIcs}, pages 63:1--63:5. Schloss Dagstuhl -
  Leibniz-Zentrum f{\"{u}}r Informatik, 2017.

\bibitem{berlekamp2000dots}
Elwyn~R. Berlekamp.
\newblock {\em The Dots and Boxes Game: Sophisticated Child's Play}.
\newblock AK Peters/CRC Press, 2000.

\bibitem{winning_ways}
Elwyn~R. Berlekamp, John~H. Conway, and Richard~K. Guy.
\newblock Chapter 16: Dots-and-boxes.
\newblock In {\em Winning Ways for your Mathematical Plays}, volume~3, pages
  541--584. A K Peters/CRC Press, 2nd edition, 2003.

\bibitem{Bodlaender92}
Hans~L. Bodlaender.
\newblock On disjoint cycles.
\newblock In {\em Graph-Theoretic Concepts in Computer Science}, pages
  230--238, Berlin, Heidelberg, 1992. Springer Berlin Heidelberg.

\bibitem{Caprara2002}
Alberto Caprara and Romeo Rizzi.
\newblock Packing triangles in bounded degree graphs.
\newblock {\em Information Processing Letters}, 84(4):175--180, 2002.

\bibitem{Epp}
David Eppstein.
\newblock Computational complexity of games and puzzles.
\newblock Last accessed on 14/02/2020.
\newblock URL: \url{https://www.ics.uci.edu/~eppstein/cgt/hard.html}.

\bibitem{Garey79}
Michael~R. Garey and David~S. Johnson.
\newblock {\em Computers and Intractability: A Guide to the Theory of
  NP-Completeness}.
\newblock W. H. Freeman \& Co., USA, 1979.

\bibitem{Lichtenstein82}
David Lichtenstein.
\newblock Planar formulae and their uses.
\newblock {\em SIAM Journal on Computing}, 11(2):329--343, 1982.

\bibitem{DP-first}
Sian Zelbo.
\newblock Dots and polygons game.
\newblock Last accessed on 14/02/2020.
\newblock URL:
  \url{http://www.1001mathproblems.com/2015/03/for-printable-game-boards-click-here.html}.

\end{thebibliography}

\clearpage

\end{document}